\documentclass[runningheads, cleverref,orivec]{llncs}
\usepackage{amsmath}
\usepackage{float}
\usepackage{graphicx}
\usepackage{amssymb}
\usepackage{hyperref}
\usepackage{mathtools}
\usepackage{svg}
\usepackage{algorithm}
\usepackage[noend]{algpseudocode}
\usepackage{booktabs}
\usepackage{bbm}
\usepackage[thinlines]{easytable}

\usepackage{tikz}
\usetikzlibrary{automata, fit, positioning, matrix, arrows.meta, decorations.pathmorphing, shapes, shapes.multipart, chains,shadows.blur}

\usepackage{tikz-qtree}

\usepackage{tikz-dependency}

\tikzset{
-{Latex[length=1.8mm, width=1.4mm]}, 
every initial by arrow/.append style={anchor/.append style={shape=coordinate}},
every state/.style={thick, fill=white, minimum size=8mm},
every loop/.append style=-{Latex[length=1.8mm, width=1.4mm]},
initial text=,
initial distance=0.5cm,
line width=0.7pt,
node distance=18mm,
minimum size=4mm,
}

\definecolor{silver}{rgb}{0.75, 0.75, 0.75}
\definecolor{slategray}{rgb}{0.44, 0.5, 0.56}

\newcommand{\mymacro}[1]{{#1}}
\newcommand{\fapprox}[1]{\widehat{#1}}

\newcommand{\defn}[1]{\textbf{#1}}
\newcommand{\blackqed}{\hfill$\blacksquare$}

\usepackage{todonotes}
\definecolor{mintgreen}{RGB}{152, 255, 152}
\makeatletter
\newcommand*\iftodonotes{\if@todonotes@disabled\expandafter\@secondoftwo\else\expandafter\@firstoftwo\fi}  %
\makeatother

\usepackage{colortbl}

\definecolor{ETHBlue}{RGB}{33,92,175}	
\definecolor{ETHGreen}{RGB}{98,115,19}		
\definecolor{ETHPurpleDark}{RGB}{140,10,89}	
\definecolor{ETHPurple}{RGB}{163,7,116}	
\definecolor{ETHGray}{RGB}{111,111,111}	
\definecolor{ETHRed}{RGB}{183,53,45}	
\definecolor{ETHPetrol}{RGB}{0,120,148}	
\definecolor{ETHBronze}{RGB}{142,103,19}	

\colorlet{ETHdarkblue}{ETHBlue!80!black}
\colorlet{ETHdarkgreen}{ETHGreen!80!black}
\colorlet{ETHpink}{ETHPurple}
\colorlet{ETHgray}{ETHGray}
\colorlet{ETHred}{ETHRed}
\colorlet{ETHgreenblue}{ETHPetrol}
\colorlet{ETHbrown}{ETHBronze}

\colorlet{MacroColor}{ETHPetrol}
\colorlet{MACROCOLOR}{MacroColor}

\newtheorem{observation}[lemma]{Observation}
\spnewtheorem{nclaim}{Claim}{\itshape}{\rmfamily}

\usepackage{algorithm}
\usepackage[noend]{algpseudocode}  
\algrenewcommand\algorithmicindent{1.0em}%
\newcommand{\rightcomment}[1]{{\color{gray} \(\triangleright\){\footnotesize\textit{#1}}}}
\algrenewcommand{\algorithmiccomment}[1]{\hfill \rightcomment{#1}}  
\algnewcommand{\LinesComment}[1]{\State {\color{black!50!green}\rightcomment{\parbox[t]{.95\linewidth-\leftmargin-\widthof{\(\triangleright\) }}{#1}}}}
\algnewcommand{\LineComment}[1]{\State {\color{black!50!green}\smaller \(\triangleright\) \parbox[t]{\linewidth-\leftmargin-\widthof{\(\triangleright\) }}{\it #1}\smallskip}} 
\algnewcommand{\InlineComment}[1]{\hfill {\color{black!50!green}\(\triangleright\) {\scriptsize \it #1}}}
\algrenewcommand\algorithmicindent{1.0em}%
\algrenewcommand\alglinenumber[1]{{\tiny\color{black!50}#1.}\hspace{-2pt}}
\newcommand{\algorithmicfunc}[1]{\textbf{def} {#1}:}
\algdef{SE}[FUNC]{Func}{EndFunc}[1]{\algorithmicfunc{#1}}{}
\makeatletter
\ifthenelse{\equal{\ALG@noend}{t}}%
  {\algtext*{EndFunc}}
  {}%
\makeatother

\newcommand{\range}[1]{{\mymacro{\mathcal{R}\!\left(#1\right)}}}

\newcommand{\field}{{\mymacro{ \mathbb{F}}}}

\newcommand{\B}{{\mymacro{ \mathbb{B}}}}

\newcommand{\norm}[1]{{\mymacro{ \left\lVert #1 \right\rVert}}}

\newcommand{\pclass}{{\mymacro{\textup{\textsf{P}}}}}

\newcommand{\pspace}{{\mymacro{\textup{\textsf{PSPACE}}}}}
\newcommand{\npspace}{{\mymacro{\textup{\textsf{NPSPACE}}}}}

\newcommand{\transmonoid}[1]{{\mymacro{\mathbb{S}(#1)}}}

\newcommand{\Boolean}{{\mymacro{\mathbb{B}}}}

\newcommand{\rank}[1]{{\mymacro{\mathrm{rank}(#1)}}}

\newcommand{\alphabet}{{\mymacro{ \Sigma}}}

\newcommand{\kleene}[1]{{\mymacro{#1^*}}}

\def\1{\mathbf{1}}

\def\eps{{\mymacro{ \varepsilon}}}

\newcommand{\str}{{\mymacro{\boldsymbol{w}}}}

\newcommand{\strx}{{\mymacro{\boldsymbol{x}}}}
\newcommand{\stry}{{\mymacro{\boldsymbol{y}}}}
\newcommand{\strz}{{\mymacro{\boldsymbol{z}}}}
\newcommand{\sym}{{\mymacro{w}}}

\newcommand{\syma}{{\mymacro{a}}}
\newcommand{\symb}{{\mymacro{b}}}
\newcommand{\symc}{{\mymacro{c}}}

\newcommand{\set}[1]{{\mymacro{\left\{ #1 \right\}}}}

\newcommand{\setcomplement}[1]{{\mymacro{#1^{\textsf{c}}}}}

\newcommand{\N}{{\mymacro{ \mathbb{N}}}}

\newcommand{\nstates}{{\mymacro{ |\states|}}}

\newcommand{\nsymbols}{{\mymacro{ |\alphabet|}}}

\newcommand{\graph}{\mymacro{ \mathcal{G}}}

\newcommand{\automaton}{{\mymacro{ \mathcal{A}}}}

\newcommand{\detautomatonprime}{{\mymacro{ {\automaton'}_{\textsc{det}}}}}
\newcommand{\mooreautomatonprime}{{\mymacro{ {\automaton'}_{\textsc{moore}}}}}
\newcommand{\ssautomaton}{{\mymacro{ \automaton_{\textsc{ss}}}}}
\newcommand{\ssautomatonprime}{{\mymacro{ \automaton'_{\textsc{ss}}}}}
\newcommand{\ssstates}{{\mymacro{ \states_{\textsc{ss}}}}}
\newcommand{\sstrans}{{\mymacro{ \trans_{\textsc{ss}}}}}
\newcommand{\ssfinal}{{\mymacro{ \final_{\textsc{ss}}}}}

\newcommand{\stateq}{{\mymacro{ q}}}

\newcommand{\states}{{\mymacro{ Q}}}

\newcommand{\trans}{{\mymacro{ \delta}}}

\newcommand{\apath}{{\mymacro{ \boldsymbol \pi}}}

\newcommand{\initial}{{\mymacro{ I}}}

\newcommand{\final}{{\mymacro{ F}}}

\newcommand{\languagefsa}[1]{{\mymacro{\mathcal{L} ( #1 )}}}

\newcommand{\mooreautomaton}{{\mymacro{ \automaton_{\textsc{moore}}}}}
\newcommand{\mfautomaton}{{\mymacro{ \automaton_{\textsc{mf}}}}}
\newcommand{\universalautomaton}{{\mymacro{ \automaton_{\textsc{univ}}}}}

\newcommand{\fsatuple}{{\mymacro{ \left( \alphabet, \states, \initial, \final, \trans \right)}}}

\newcommand{\edgenoweight}[3]{#1 \xrightarrow{#2} #3}

\newcommand{\powerstate}{{\mymacro{ \mathcal{Q}}}}
\newcommand{\initpowerstate}{{\mymacro{ \mathcal{Q}_{\initial}}}}

\def\gJ{{{\mymacro{ \mathcal{J}}}}}

\def\gR{{{\mymacro{ \mathcal{R}}}}}

\def\sC{{{\mymacro{ \mathcal{C}}}}}

\def\sF{{{\mymacro{ \mathcal{F}}}}}

\def\sI{{{\mymacro{ \mathcal{I}}}}}
\def\sJ{{{\mymacro{ \mathcal{J}}}}}

\def\sT{{{\mymacro{ \mathcal{T}}}}}

\newcommand{\bigO}[1]{{\mymacro{ \mathcal{O}\left(#1\right)}}}

\usepackage{cleveref}
\crefname{section}{Section}{Sections}
\crefname{table}{Table}{Tables}
\crefname{figure}{Figure}{Figures}
\crefname{algorithm}{Algorithm}{Algorithms}
\crefname{equation}{Eq.}{Eqs.}
\crefname{example}{Example}{Examples}
\crefname{fact}{Fact}{Facts}
\crefname{nclaim}{Claim}{Claims}
\crefname{appendix}{Appendix}{Appendices}
\crefname{theorem}{Theorem}{Theorems}
\crefname{reTheorem}{Theorem}{Theorems}
\crefname{aquestion}{Question}{Questions}
\crefname{assumption}{Assumption}{Assumptions}
\crefname{lemma}{Lemma}{Lemmas}
\crefname{reLemma}{Lemma}{Lemmas}
\crefname{proposition}{Proposition}{Propositions}
\crefname{chapter}{Chapter}{Chapters}
\crefname{line}{line}{lines}
\crefname{principle}{Principle}{Principles}
\crefname{definition}{Definition}{Definitions}
\crefname{corollary}{Corollary}{Corollaries}
\crefname{Exercise}{Exercise}{Exercises}
\crefname{observation}{Observation}{Observations}

\usepackage[T1]{fontenc}
\usepackage{appendix}

\usepackage{color}

\begin{document}

\title{A Close Analysis of the Subset Construction} 


\author{Ivan Baburin \and Ryan Cotterell}

\institute{Department of Computer Science,
  ETH Z{\"u}rich, Switzerland
\email{ivan.baburin@inf.ethz.ch}\quad \email{ryan.cotterell@inf.ethz.ch}}

\maketitle

\begin{abstract}
  Given a nondeterministic finite-state automaton (NFA), we aim to estimate the size of an equivalent deterministic finite-state automaton (DFA).
  We demonstrate that computing the state complexity of an NFA within polynomial precision is $\pspace$-hard. 
  Furthermore, we also demonstrate that it is $\pspace$-hard to decide whether the classical subset construction will yield an equivalent DFA with an exponential increase in the number of states. 
  This result implies that making any \textit{a priori} estimate of the running time of the subset construction is inherently difficult. To address this, and to enable forecasting of such exponential blow-up in certain special cases, we introduce the notion of subset complexity, which provides an upper bound on the size of the DFA produced by the subset construction.
We show that the subset complexity can be efficiently bounded above using the cyclicity and rank of the transition matrices of the NFA. This yields a sufficient condition for identifying NFAs that can be efficiently determinized via the subset construction.\looseness=-1

\end{abstract}


\section{Introduction}\label{sec:introduction}
The equivalence in expressivity between nondeterministic finite-state automata (NFAs) and deterministic finite-state automata (DFAs) was one of the first foundational results in the theory of finite-state automata. The original proof by Rabin and Scott~\cite{RabinPowerset} introduced the powerset construction, a general algorithm for converting an arbitrary NFA into an equivalent DFA.
Despite its simplicity and pedagogical value, a major drawback of the powerset construction is that it yields a DFA with exactly $2^n$ states, where $n$ is the number of states in the original NFA. 
In general, the equivalent DFA produced by the powerset construction is \emph{not} minimal and can, thus, often be reduced using standard finite-state minimization techniques~\cite{brzozowski-full}. Moreover, in certain pathological cases, the minimal DFA is \emph{exponentially} smaller than the DFA returned by the powerset construction—meaning the algorithm may perform an exponential amount of unnecessary computation.
A trivial example illustrates this phenomenon: consider an arbitrary $n$-state NFA with \emph{no} final states. In this case, the minimal DFA is the 1-state automaton that accepts the empty language. 
However, the exponential blow-up is unavoidable in the worst case. Specifically, \cite{moore1971bounds} showed the existence of an $n$-state NFA, known as the Moore automaton,\footnote{See \cref{sec:difficulties} for further exposition.} 
for which the \emph{minimal} equivalent DFA requires \emph{exactly} $2^n$ states, i.e., the exponentially large DFA output by the powerset construction is already minimal.
Subsequent research identified other examples of NFAs for which the powerset construction yields the minimal DFA~\cite{BORDIHN20093209}.\looseness=-1

A complete theory of determinization for finite-state automata ought to provide a principled understanding of precisely when the minimal DFA  equivalent to an NFA incurs an exponential blow-up in the number of states, and when it does not. To this end, several alternative determinization algorithms have been proposed, such as the classic subset construction \cite{allauzen2007openfst} and Brzozowski's algorithm~\cite{brzozowski}.
In addition, it has been shown that certain properties of NFAs can be decided \emph{without} performing explicit determinization despite the fact that many naive algorithms for those tasks construct an equivalent DFA.
For example, NFA equivalence and universality can be tested directly~\cite{hacking-universality}. Although these problems are often $\pspace$-hard to decide—universality being a prime example—modern techniques frequently outperform approaches that rely on full determinization~\cite{antichains-tests}.
Nevertheless, our overall understanding remains fragmented, and a comprehensive, unifying theory of determinization continues to elude us.

In this paper, we present new results on forecasting the size of an NFA's equivalent DFA. 
Our primary focus is a novel analysis of the subset construction, a dynamic modification of the powerset construction, which only keeps accessible states. 
The subset construction is the most basic and widely used determinization procedure, implemented in many automata toolkits~\cite{allauzen2007openfst}. It is also relatively well understood---the cases where it results in the minimal DFA were fully characterized by Brzozowski~\cite{brzozowski-full}.
However, Brzozowski's characterization still does not tell us whether an NFA's equivalent minimal DFA blows up exponentially.
In that context, as our first contribution, in \cref{sec:difficulties}, we prove that no algorithm can efficiently predict whether the minimal DFA equivalent to a given NFA will exhibit an exponential increase in the number of states.
Moreover, even approximating such a forecast in a coarse-grained manner is $\pspace$-hard  (\cref{prop:pspace}).
More interestingly, we also show that the same holds when forecasting whether the subset construction will produce an equivalent DFA, non-minimal in general, with exponentially more states (\cref{prop:subset-pspace}).
As our second contribution, we turn to analyzing the size of the DFA produced by the subset construction.
Simply restricting nondeterminism does not prevent exponential blow-up when performing the subset construction: unambiguous NFAs and automata with multiple initial states (but otherwise deterministic structure) can still yield exponentially larger DFAs~\cite{ambiguous-NFA}. 
Prior attempts to bound the blow-up have examined special classes such as homogeneous NFAs~\cite{subsethomogenous} and employed algebraic tools like monoid complexity~\cite{onthefly}, but no general bound for subset construction has emerged.
Thus, in \cref{sec:subset-complexity}, we introduce a simple yet general measure called subset complexity, which provides an upper bound on the size of the DFA produced by the subset construction (\cref{thm:decomposition-bound}), and thereby a bound on the size of the minimal DFA. 
Finally, in \cref{sec:computational}, we propose an efficient estimate for computing the subset complexity in practice (\cref{prop:all-but-one}).
Taken together, we hope these results contribute to a deeper theoretical understanding of the efficiency of determinizing NFAs.\looseness=-1

\section{Background: Automata, Determinization and Monoids}\label{sec:background}
We start the technical portion of the paper by briefly reciting necessary definitions concerning finite-state automata; for a more comprehensive introduction, we refer to \cite{automata-book}.
An \defn{alphabet} $\alphabet$ is a finite, non-empty set of symbols.
With $\kleene{\alphabet}$, we denote the set of all finite words over $\alphabet$.
The symbol $\eps \in \kleene{\alphabet}$ denotes the empty word, i.e., the unique word of length 0.
Two words are placed adjacent to one another, e.g., $\strx \stry$,  denotes the concatenation of those words.
\begin{definition}
    A \defn{finite-state automaton (FSA)} $\automaton = \fsatuple$ is a quintuple where $\alphabet$ is an alphabet, $\states$ is a finite set of states, $\initial \subseteq \states$ is a set of the initial states, $\final \subseteq \states$ is a set of accepting states and $\trans \subseteq \states \times (\alphabet \cup \set{\eps}) \times \states$ is a finite multiset of transitions.
\end{definition}
Note that a transition labeled with $\eps$ can be performed at any point.
The \defn{size} of a finite-state automaton $\automaton =  \fsatuple$ is its number of states; we use the cardinality operator $| \automaton | = \nstates$ to denote the size. 
For the remainder of this article, we assume that all automata are $\eps$-free.
This assumption is without loss of generality as there exist efficient procedures for removing $\eps$-transitions \cite{allauzen2007openfst}. 

\begin{definition}\label{def:path}
    A \defn{path} $\apath = \edgenoweight{\stateq_0}{\sym_1}{\stateq_1} \edgenoweight{}{\sym_2}{\ldots} \edgenoweight{}{\sym_{N}}{\stateq_N}$ is a labeled sequence of states where each label $\sym_n \in \alphabet \cup \{\eps\}$ is a symbol from the alphabet and each $\stateq_n \in \states$ is a state in the FSA.
    We will refer to the concatenation of symbols along the path $\sym_1 \cdots \sym_{N}$ as the path's \defn{yield}. 
\end{definition}

We say that a word $\str \in \kleene{\alphabet}$ is \defn{recognized} by the FSA $\automaton$ if there exists a path from some initial state $\stateq_I \in \initial$ to some final state $\stateq_F \in \final$ with yield $\str$. 
Moreover, we denote with $\languagefsa{\automaton} \subseteq \kleene{\alphabet}$ the \defn{language} (set of all words) recognized by $\automaton$. 
Two FSA $\automaton$ and $\automaton'$ are called \defn{equivalent} if they accept the same language, i.e., if $\languagefsa{\automaton} = \languagefsa{\automaton'}$. 

\begin{definition}\label{def:deterministic}
An FSA \textup{$\automaton = \fsatuple$} is \defn{deterministic} if and only if the initial state is unique, i.e., $|\initial| = 1$, and for every $(\stateq, \sym) \in \states \times \alphabet$ there is exactly one $\stateq' \in \states$ such that $(\stateq, \sym, \stateq') \in \trans$, i.e., the relation corresponds to a total function of type $\states \times \alphabet \rightarrow \states$.
An FSA that is not deterministic is called \defn{nondeterministic}.
\end{definition}
As a shorthand, for any $\str \in \kleene{\alphabet}$, we write $\trans(\stateq, \str) = \trans(\cdots \trans( \trans(\stateq, \sym_1), \sym_2) \cdots, \sym_N)$ to denote an iterative application of the transition function to the word $\str= \sym_1 \cdots \sym_N$.
In general, DFAs are both easier to implement and to use---many properties, e.g., equivalence to another DFA, are hard to decide for general NFAs, but can be decided for deterministic automata in polynomial time \cite{pspace-survey}. 

\begin{definition}
    The \defn{reverse} FSA of an automaton $\automaton$, marked by $\automaton^R$, is an FSA obtained by reversing the direction of all transitions in $\automaton$ and swapping the initial and final states. 
    We call FSA $\automaton$ \defn{co-deterministic} if $\automaton^R$ is deterministic.
\end{definition}

\begin{definition}\label{def:trim}
  A state $\stateq \in \states$ is called \defn{accessible} if there exists a path from $\initial$ to $\stateq$, and \defn{co-accessible} if there exists a path from $\stateq$ to $\final$. An automaton where all states are both accessible and co-accessible is called \defn{trim}.
\end{definition}
Transforming a nondeterministic automaton into an equivalent deterministic one is known as \defn{determinization}.
The most commonly implemented practical algorithm for determinization is the \defn{subset construction} \cite{allauzen2007openfst}, a dynamic version of powerset construction which only keeps track of accessible states; hence, the states in the DFA returned by the subset construction are always a subset of the states in the DFA returned by the powerset construction. \cref{alg:subsetconstruction} provides the pseudocode for an example implementation of subset construction.
The algorithm roughly works as follows.
Given an input NFA $\automaton = \fsatuple$, it explores all accessible subsets of states $\states$ by following the labeled transitions in $\trans$. 
Each new accessible state $\powerstate \subseteq \states$ will be placed on the stack, and for every subset of states $\states$ on the stack, we iteratively search for all of its direct neighbors and place them on the stack until the stack is empty. Algebraically, this can be seen as multiplying the characteristic vector for the initial state $\powerstate_\initial$ with all possible sequences of transition matrices and terminating as soon as no new states can be produced. 
We will refer to the DFA $\ssautomaton$ obtained by applying subset construction on $\automaton$ as the \defn{subset automaton} for $\automaton$.

\begin{algorithm}[t]
\caption{\textbf{(Subset construction)}}\label{alg:subsetconstruction}

\begin{algorithmic}
\Ensure $\automaton = \fsatuple$ 
\State $\initpowerstate \gets \{ \stateq \ | \ \stateq \in \initial \}$
\State $\ssautomaton \gets (\alphabet, \ssstates, \powerstate_{\initial}, \ssfinal, \sstrans)$
\State $\mathsf{stack} \gets \powerstate_\initial$
\State $\ssstates \gets \{ \powerstate_\initial \}$
\While{$|\mathsf{stack}| > 0$}
    \State \textbf{pop} $\powerstate$ from the $\mathsf{stack}$
    \ForAll{$\sym \in \alphabet $}
        \State $\powerstate' \gets \{\stateq' \ | \ (\stateq, \sym, \stateq') \in \trans, \ \stateq \in \powerstate \}$
        \State $\sstrans \gets \sstrans \cup \{\edgenoweight{\powerstate}{\sym}{\powerstate'} \}$
        \If{$\powerstate' \notin \ssstates$}
            \State $\ssstates \gets \ssstates \cup \{ \powerstate' \}$
            \State \textbf{push} $\powerstate'$ on the $\mathsf{stack}$
        \EndIf
    \EndFor
\EndWhile
\State $\ssfinal \gets \{ \powerstate \in \ssstates \ | \ \powerstate \cap \final \neq \emptyset \}$
\State \textbf{return} $\gets \ssautomaton$
\end{algorithmic}
\end{algorithm}

\begin{definition}\label{def:state-complexity}
    Given a finite-state automaton $\automaton$, we call the number of states in its \defn{minimal} equivalent deterministic automaton its \defn{state complexity}.
\end{definition}
Another way of characterizing FSA is by using the algebraic approach and studying the semigroups defined by the underlying transition matrices.

\begin{definition}
    Consider an FSA $\automaton = \fsatuple$ and, for each $\sym \in \alphabet$, define a binary relation ${T}^{(\sym)}$ over $\states$ with $(\stateq, \stateq') \in T^{(\sym)} \Leftrightarrow (\stateq, \sym, \stateq') \in \trans $ for all $\stateq, \stateq' \in \states$. A \defn{transition monoid} $\transmonoid{\automaton}$ is a binary relation monoid over $\states$ generated by $ \{ T^{(\sym)} \ | \ \sym \in \alphabet \}$ and closed under the relation composition operator $\circ$ with the identity relation \textup{$\mathrm{id}$}.\looseness=-1
\end{definition}

It is often more convenient to think of $\transmonoid{\automaton}$ as a monoid of $\nstates \times \nstates$ matrices over the Boolean semifield $\Boolean \coloneqq \langle \{ 0, 1 \}, \vee, \wedge, 0, 1 \rangle$ closed under Boolean matrix multiplication with the identity element $\sI$, such that:
\begin{equation}
    T^{(\sym)}_{i, j} = 1 \Longleftrightarrow (\stateq_i, \stateq_j) \in T^{(\sym)}.
\end{equation}
Notice that the monoid $\transmonoid{\automaton}$ is, by construction, closely related to the regular language $\languagefsa{\automaton}$.
This duality has been extensively studied; for a more detailed overview, we refer to \cite{pin1997syntactic}.

\section{The Complexity of Determinization}\label{sec:difficulties}

As mentioned in \Cref{sec:introduction}, the problem of converting an NFA into an equivalent DFA is difficult because the state complexity of many NFAs can be exponentially large. 
Two classical examples demonstrating this phenomenon are the $n$-state Moore automaton \cite{moore1971bounds} and the $n$-state Meyer--Fischer automaton \cite{meyerfischer}, both of which exhibit a state complexity of $2^n$; see \cref{fig:Moore-NFA}. 

\begin{example}
    The blow-up occurring in a Meyer--Fischer automaton $\mfautomaton$ has a particularly intuitive explanation. 
    Let ${| \str |}_{\syma}$ denote the number of times the symbol $\syma$ appears in a word $\str \in  \kleene{\{\syma, \symb \}}$.
    Then, the language accepted by an $n$-state $\mfautomaton$ can be described as
    \begin{equation}
        \languagefsa{\mfautomaton} = \{ \strx \stry  \ | \ \strx \in \{ \varepsilon \} \cup (\kleene{\{\syma, \symb\}}\symb), \; \stry \in \kleene{\{\syma, \symb\}} \text{ and } {|\stry|}_{\syma} \equiv 0 \pmod{n}\}. 
    \end{equation}
    For every word $\str \in  \languagefsa{\mfautomaton}$, the nondeterministic Meyer--Fischer automaton guesses a prefix that is either empty or ends with $\symb$, such that the number of symbols $\syma$ in the remaining suffix is a multiple of $n$. 
    To accomplish that, any deterministic counterpart of $\mfautomaton$ needs to keep track of all parities modulo $n$ for all occurrences of $\symb$ in $\str$, i.e., exactly $2^n$ possible options.
\end{example}

We emphasize that not all NFAs will exhibit exponential blow-up.
For example, as mentioned in \Cref{sec:introduction}, an NFA with no accepting states can be converted to an equivalent DFA with a single state accepting the empty language. 
Nevertheless, due to the possibility of an exponential blow-up, it immediately follows that there does not exist a polynomial-time (in the size of the NFA) algorithm that outputs the equivalent minimal DFA when given an NFA as input.

The next logical step would be to instead look for an efficient \emph{parameterized} algorithm to find a minimal equivalent DFA, i.e., a polynomial-time algorithm parameterized by the state complexity of the NFA. 
Somewhat surprisingly, this is likely impossible as well.
To see this, consider the following one-state deterministic automaton 
\begin{equation}\label{eq:universal-dfa}
    \universalautomaton \coloneqq (\{\syma, \symb\}, \{ \stateq \}, \{\stateq \}, \{\stateq \}, \{(\stateq, \syma, \stateq), (\stateq, \symb, \stateq) \}),
\end{equation}
which is minimal and accepts the \defn{universal} language $\languagefsa{\universalautomaton} = \kleene{\{\syma, \symb\}}$, i.e., all words over a binary alphabet. 
The subsequent result demonstrates that the minimization becomes intractable even for this specific case.

\begin{figure}[t!]
    \centering
    \begin{tikzpicture}[->,>=stealth',shorten >=0.5pt,auto,node distance=15mm, scale=0.8, transform shape]
        \tikzstyle{every state}=[fill=none,draw=black,text=black,minimum size=0pt,inner sep=2.5pt]
            \node[state, initial] (q1) {$q_1$};
            \node[state, right of=q1] (q2) {$q_2$};
            \node[state, right of=q2] (q3) {$q_3$};
            \node[fill=white, right of=q3] (q?) {$\ldots$};
            \node[state, accepting, right of=q?] (qn) {$q_n$};
            \draw (q1) edge[loop above] node{$b$} (q1)
                  (q1) edge[above] node{$a$} (q2)
                  (q2) edge[above] node{$a, b$} (q3)
                  (q3) edge[above] node{$a, b$} (q?)
                  (q?) edge[above] node{$a, b$} (qn)
                  (qn) edge[bend right=1.7cm, above] node{$a$} (q2)
                  (qn) edge[bend left=1.5cm, below] node{$a$} (q1);
    \end{tikzpicture}
    \hfill
    \begin{tikzpicture}[->,>=stealth',shorten >=0.5pt,auto,node distance=15mm, scale=0.8, transform shape]
        \tikzstyle{every state}=[fill=none,draw=black,text=black,minimum size=0pt,inner sep=2.5pt]
            \node[state, accepting, initial] (q1) {$q_1$};
            \node[state, above right = {12mm} and {3mm} of q1] (q2) {$q_2$};
            \node[state, above right of=q2] (q3) {$q_3$};
            \node[state, right of=q3] (q4) {$q_4$};
            \node[fill=white, below right of=q4] (q?) {$\ldots$};
            \node[state, right = {43mm} of q1] (qn) {$q_n$};

            \draw (q1) edge[in=85,out=115,loop, above] node{$b$} (q1)
                  (q2) edge[loop above] node{$b$} (q2)
                  (q3) edge[loop above] node{$b$} (q3)
                  (q4) edge[loop above] node{$b$} (q4)
                  (qn) edge[in=50,out=80,loop, above] node{$b$} (qn)
                  (q1) edge[above left] node{$a$} (q2)
                  (q2) edge[above left] node{$a$} (q3)
                  (q3) edge[above] node{$a$} (q4)
                  (q4) edge[above right] node{$a$} (q?)
                  (q?) edge[above right] node{$a$} (qn)
                  (qn) edge[above] node{$a, b$} (q1)
                  (q2) edge[bend left=1cm,left] node{$b$} (q1)
                  (q3) edge[bend left=1cm,left] node{$b$} (q1)
                  (q4) edge[bend left=1cm,above] node{$b$} (q1);
        \end{tikzpicture}
    \caption{An $n$-state Moore automaton $\mooreautomaton$ (left) and an $n$-state Meyer--Fischer automaton $\mfautomaton$ (right). Both of them are nondeterministic with state complexity $2^n$ for every $n \geq 2$.}
    \label{fig:Moore-NFA}
\end{figure}

\begin{theorem}[NFA Universality \cite{pspace-survey}]\label{thm:universality}
    Given an $n$-state NFA $\automaton$, deciding whether \textup{$\languagefsa{\automaton} = \languagefsa{\universalautomaton}$} is \textup{$\pspace$}-complete.
\end{theorem}
Under the assumption that $\pspace \neq \pclass$, we conclude that no parameterized polynomial algorithm for determinization exists---otherwise, we could decide universality by simply executing this algorithm and terminating it after some fixed polynomial number of steps. 
For the same reason, there is no efficient algorithm to find an almost minimal equivalent DFA because DFAs can be minimized efficiently. 
Furthermore, this result also implies that it is likely impossible to devise an algorithm that efficiently determines the state complexity of an NFA.
\begin{corollary}\label{cor:state-complexity}
Determining the state complexity of a given NFA, i.e., the number of states in a minimal equivalent DFA, is \textup{$\pspace$}-hard.
\end{corollary}
\begin{proof}
     This follows immediately from \cref{thm:universality}.
     Note that the state complexity of \Cref{eq:universal-dfa} is $1$.
     An NFA with state complexity of $1$ must either be universal or accept the empty language. 
     However, checking the emptiness of an NFA is a simple graph connectivity problem, which can be solved in linear time.
     Thus, given that we can efficiently determine whether an NFA accepts only the empty language, we have obtained a polynomial reduction from determining universality to determining state complexity. \blackqed
\end{proof}
We now relax the original question and ask whether it is possible to find some approximation to the state complexity of a given NFA in polynomial time. Perhaps even more surprisingly, this question remains equally hard as the previous one, even if we allow for arbitrary polynomial slackness in our approximation.
Let us denote with $\mathfrak{A}$ the set of all FSA.
To formalize the notion of approximation, we will require an additional definition.

\begin{definition}
    Let $\automaton$ be an arbitrary NFA and let $f \colon \N \rightarrow \N$ be some function satisfying $f(i) \geq i$ for all $i \in \N$. We say that metric $\sF \colon \mathfrak{A} \rightarrow \N$ causes an $f$-\defn{blow-up} for $\automaton$ if\looseness=-1
    \begin{equation}
        \sF(\automaton) > f(|\automaton|).
    \end{equation}
    Furthermore, we say that $\fapprox{m}$ is an $f$-\textbf{approximation} of $\sF$ on $\automaton$ if \begin{equation}
        \sF(\automaton) \leq \fapprox{m} \leq f(\sF(\automaton) + |\automaton|),
    \end{equation}
    i.e., $\fapprox{m}$ is an upper bound with slackness specified by $f$.
\end{definition}
For our purposes, only two particular choices for $\sF$ are relevant: the state complexity of $\automaton$ and the size of automaton produced by the subset construction applied to $\automaton$. \looseness=-1
\begin{proposition}\label{prop:pspace}
   For any polynomial $p$, it is $\pspace$-hard both to compute a $p$-approximation to the state complexity and to decide whether a $p$-blow-up will occur.\looseness=-1
\end{proposition}
\begin{proof}
    To show $\pspace$-hardness, we construct a polynomial reduction from the NFA universality problem. 
    Let $n = |\automaton|$ and $m$ be the (unknown) state complexity of $\automaton$.
    Further, assume that there exists an algorithm that produces a $p$-approximation for $m$ (the proof for $p$-blow-up is analogous). 
    We show that we could use such an algorithm to decide whether $\automaton$ accepts the universal language, which is $\pspace$-hard by \Cref{thm:universality}.
    
    Without loss of generality, we assume that $\alphabet = \{\syma, \symb\}$.
    Next, we construct an automaton $\automaton'$ over the alphabet $\alphabet' = \{\syma, \symb, \# \}$ with $2(n+1)$ states using a copy of $\automaton$, a copy of $\mooreautomaton$ and two copies of $\universalautomaton$ as shown in \cref{fig:minimal-pspace}. 
    Specifically, $\automaton'$ is the union of the following two NFAs:\looseness=-1
    \begin{enumerate}
        \item We construct the first automaton as follows.
        Let $\universalautomaton$ be the 1-state minimal DFA that accepts $\kleene{\alphabet}$, and let $\mooreautomaton$ be the $n$-state Moore automaton. 
        Then, concatenate $\universalautomaton$ to  $\mooreautomaton$ by adding a $\#$-labeled transition from unique state in $\universalautomaton$ with the initial state of $\mooreautomaton$. 
        Then, the initial state of $\mooreautomaton$ no longer serves as an initial state, and the single state in $\universalautomaton$ no longer serves as a final state.
        \item We construct the second automaton as follows.
        Let $\universalautomaton$ be another copy of the 1-state minimal DFA that accepts $\kleene{\alphabet}$. 
         Then, concatenate $\automaton$ to  $\universalautomaton$ by adding a $\#$-labeled transition from all final states in $\automaton$ to the initial state of $\universalautomaton$. 
         The unique state in $\universalautomaton$ no longer serves as an initial state, and $\automaton$ no longer has any final states.
    \end{enumerate}
    \begin{figure}[t]
    \centering
    \scalebox{.8}{
        \begin{tikzpicture}[->,>=stealth',shorten >=0.5pt,auto,node distance=15mm,
                            semithick]
        \tikzstyle{every state}=[fill=none,draw=black,text=black,minimum size=0pt,inner sep=1pt]
        
        \node[state, initial] (p1) {$p_1$};
        \node[state, above right of = p1] (p2) {$p_2$};
        \node[state, right of = p2] (p3) {$p_3$};
        \node[state, right of = p3] (p4) {$p_4$};
        \node[state, below right of = p1] (p5) {$p_5$};
        \node[fill=white, below right =2.1cm and 0.1cm of p4] (pmain){$\automaton$};
        
        \draw (p1) edge[loop below] node{$a, b$} (p1)
              (p1) edge[above left] node{$a$} (p2)
              (p1) edge[above right] node{$b$} (p5)
              (p2) edge[above] node{$a$} (p3)
              (p3) edge[above] node{$a$} (p4)
              (p4) edge[loop below] node{$a, b$} (p4);
    
        \node[state, accepting, below right= {0.75cm} and {2cm} of p4] (pu) {$p_u$};
        \node[fill=white, below =0.5cm of pu] (pmainu) {$\universalautomaton$};
        \draw (pu) edge[loop above] node{$a, b$} (pu);
        
        \node[state, above right = 4cm and 3.8cm of p1] (q1) {$q_1$};
        \node[state, right of=q1] (q2) {$q_2$};
        \node[state, right of=q2] (q3) {$q_3$};
        \node[state, right of=q3] (q4) {$q_4$};
        \node[fill=white, right of=q4] (q?) {$\ldots$};
        \node[state, accepting, right of=q?] (qn) {$q_n$};
        \node[fill=white, below=1.1cm of qn] (qmoore) {$\mooreautomaton$};
        \draw (q1) edge[loop above] node{$b$} (q1)
              (q1) edge[above] node{$a$} (q2)
              (q2) edge[above] node{$a, b$} (q3)
              (q3) edge[above] node{$a, b$} (q4)
              (q4) edge[above] node{$a, b$} (q?)
              (q?) edge[above] node{$a, b$} (qn)
              (qn) edge[bend right=1.6cm, above] node{$a$} (q2)
              (qn) edge[bend left=1.2cm, below] node{$a$} (q1);
    
        \node[state, initial, left = 2cm of q1] (qu) {$q_u$};
        \node[fill=white, below right =0.55cm and -0.25cm  of qu] (qmooreu) {$\universalautomaton$};
        \draw (qu) edge[loop above] node{$a, b$} (qu);
    
        \draw (qu) edge[above] node{$\#$} (q1);
        \draw (p4) edge[above right] node{$\#$} (pu);
          
        \node[draw, dashed, inner xsep=6mm, inner ysep = 4mm, fit=(p1) (p2) (p3) (p4) (p5) ] {};
        \node[draw, dashed, inner xsep=4mm, inner ysep = 16mm, fit=(q1) (q2) (q3) (q4) (q?) (qn)] {};
        \node[draw, dashed, inner xsep=6mm, inner ysep = 10mm, anchor=base, fit=(qu)] {};
        \node[draw, dashed, inner xsep=2mm, inner ysep = 10mm, anchor=base, fit=(pu)] {};
        \end{tikzpicture}
    }
    \caption{Construction for automaton $\automaton'$ consisting of a union of two nondeterministic automata: the upper one is a connection of $\universalautomaton$ and $\mooreautomaton$ using symbol $\#$, and the lower one is the connection of $\automaton$ and $\universalautomaton$ using symbol $\#$. In both cases the final states of the first automaton are connected to the initial states of the second one, with the corresponding initial/final labels dropped.} 
    \label{fig:minimal-pspace}
    \end{figure}
    Clearly, the first automaton accepts the language $\kleene{\alphabet}\#\languagefsa{\mooreautomaton}$ while the second accepts the language $\languagefsa{\automaton}\#\kleene{\alphabet}$. 
    Thus, $\automaton'$ accepts $\kleene{\alphabet}\#\languagefsa{\mooreautomaton}  \cup \languagefsa{\automaton}\#\kleene{\alphabet}$.
    On the one hand, if $\automaton$ is universal, then $\languagefsa{\automaton} = \kleene{\alphabet}$, and the union of these languages is precisely $\kleene{\alphabet}\#\kleene{\alphabet}$.
    Consequently, the state complexity of $\automaton$ is $3$.
    On the other hand, if $\languagefsa{\automaton}$ is not universal, then there exists a word $\strx \in \kleene{\{a, b \}}$ such that $\strx \notin \languagefsa{\automaton}$, meaning that
    \begin{equation*}
        \bigl\{ \strx\# \stry \mid  \stry \in \languagefsa{\mooreautomaton} \bigl\} \subset \languagefsa{\automaton'} \quad \text{and} \quad \bigl\{ \strx\#\stry  \mid \stry \in  \setcomplement{\languagefsa{\mooreautomaton}}\bigl\} \cap \languagefsa{\automaton'} = \emptyset.
    \end{equation*}
    Consider the minimal deterministic automaton for $\automaton'$, which we denote with $\detautomatonprime$. Let $q$ be the state in $\detautomatonprime$ to which the path $\apath$ with the yield $\stry\#$ leads. 
    Now, observe that the right language of $q$ is  $\languagefsa{\mooreautomaton}$.
    Because $\languagefsa{\mooreautomaton}$ has a state complexity $2^{n}$, the state complexity of $\detautomatonprime$ is at least $2^n$.
    Thus, depending on whether $\automaton$ is universal or not, the state complexity of $\automaton'$ will be either $3$ or at least $2^n$. Because $2^n$ is exponentially larger than $3$, any algorithm producing a $p$-approximation for the state complexity of $\automaton'$ would necessarily decide the universality of $\automaton$. 
    An analogous proof holds for deciding $p$-blow-up.
    \blackqed
\end{proof}

\cref{prop:pspace} shows that it is $\pspace$-hard to give any polynomial-time approximation to the blow-up in an NFA's state complexity.
Thus, instead, we relativize our results with respect to a specific determinization procedure which, given an NFA, produces an equivalent DFA, albeit one that is not necessarily minimal.
In this article, we focus on the subset construction, discussed in \Cref{sec:background}.
The question of optimality for subset construction was resolved in \cite{brzozowski-full}, who presented a characterization of all NFAs for which subset construction produces a minimal DFA. 
A well-known special case of this result is given below.
\begin{theorem}[Brzozowski \cite{brzozowski}]\label{thm:brzozowski}
    If an NFA $\automaton$ is trim and co-deterministic, then the subset construction produces the minimal equivalent DFA.
\end{theorem}
Even though the subset construction will not produce the minimal DFA for most NFAs, it does have one major advantage---it produces a DFA where all states are accessible.  
A more fine-grained analysis says that, if $\automaton$ is an NFA and $\ssautomaton$ is the (possibly non-minimal) DFA produced by the subset construction, then the running time of \cref{alg:subsetconstruction} is in $\bigO{\nsymbols |\ssautomaton| {|\automaton|}^2}$ if the algorithm is implemented with well-chosen data structures \cite{subsethomogenous}. 
This reduces the question of whether it is possible to efficiently find an equivalent DFA to one of computing the size of the subset automaton.
Clearly, we can decide whether subset construction on $\automaton$ causes an $f$-blow-up by simulating $f(|\automaton|)$ steps of \Cref{alg:subsetconstruction} as each step produces a new state, or the algorithm halts.
However, the running time of this approach is not polynomial, so it is not an efficient procedure.
In other words, the question we care about is whether it is possible \textit{a priori} to find an approximation to the size of subset automaton of $\automaton$ in polynomial time?
\begin{proposition}\label{prop:subset-pspace}
    For any polynomial $p$, it is $\pspace$-hard to both compute a $p$-approximation to the size of the subset automaton or to decide whether a $p$-blow-up will occur.
\end{proposition}
\begin{proof}
    To show $\pspace$-hardness of computing a $p$-approximation, we again construct a polynomial reduction from the NFA universality problem; the $p$-blow-up case is analogous.
    Assume that we have access to an algorithm which, given any $\automaton = \fsatuple$, computes a $p$-approximation $\widehat{m}$ to the size of $\automaton$'s subset automaton.
     Let $k$ be the degree of the polynomial $p$ and without loss of generality assume that $\alphabet = \{\syma, \symb\}$. 
     Consider an $n(k+1)$-state Meyer--Fischer automaton $\mfautomaton$ over $\alphabet$.
     As with the $n(k+1)$-state Moore automaton, the $n(k+1)$-state Meyer--Fischer $\mfautomaton$ has a state complexity of $2^{n(k+1)}$.
     However, the Meyer--Fischer automaton also has a unique initial state $p_1$ which is also the unique accepting state, as shown in the top part of \cref{fig:subset-pspace}.
    \begin{nclaim}\label{claim:MF-fixpoint}
        Let $\mfautomaton$ be an $n$-state Meyer--Fischer automaton.
        Consider the state $\powerstate_1 = \{p_1, \ldots, p_n\}$ created by applying the subset construction (\Cref{alg:subsetconstruction}) to $\mfautomaton$.
        Then, all transitions leaving $\powerstate_1$ are self-loops. 
    \end{nclaim}
    Because being in the state $\powerstate_1$ corresponds to simultaneously being in the states $\{p_1, \ldots, p_n\}$ in $\mfautomaton$ (see \cref{fig:subset-pspace}), the claim follows by noting that every state in $\mfautomaton$ has at least one incoming transition with label $\syma$ and one with $\symb$.
    In other words, once you are in $\powerstate_1$, you will always transition back to all states in $\powerstate_1$. 
    Without loss of generality, assume that $\automaton$ has an explicit dead state $\stateq_{\varnothing}$, i.e., if some state $\stateq \in Q$ is missing an explicit transition for some symbol $\sym \in \Sigma$, we augment $\automaton$ with transition $\edgenoweight{\stateq}{\sym}{\stateq_{\varnothing}}$. 
    In other words, we assume that $\automaton$'s transition function is total.\looseness=-1
    \begin{nclaim}\label{claim:all-states}
        An NFA $\automaton$ is universal if and only if every state $\powerstate$ in $\ssautomaton$, the DFA produced by the subset construction (\Cref{alg:subsetconstruction}), contains a final state from $\automaton$.
    \end{nclaim}
    To show the left implication, observe that if all states generated by applying the subset construction contain a final state, then all states in $\ssautomaton$ are accepting by the design of \cref{alg:subsetconstruction}; hence, $\ssautomaton$'s language is universal. 
    To show the right implication, observe that, if there exists a state $\powerstate$ generated by \Cref{alg:subsetconstruction} that does \emph{not} contain a final state, then, because all states in $\ssautomaton$ are accessible, there exists some path $\apath$ with yield $\stry$ from $\powerstate_\initial$ to $\powerstate$, meaning that $\stry \notin \languagefsa{\ssautomaton}$. 
    \begin{figure}[t]
    \centering
    \scalebox{.8}{
    \begin{tikzpicture}[->,>=stealth',shorten >=0.5pt,auto,node distance=15mm,
                        semithick]
    \tikzstyle{every state}=[fill=none,draw=black,text=black,minimum size=0pt,inner sep=1pt]
    
    \node[state, initial] (p1) {$q_1$};
    \node[state, above right of = p1] (p2) {$q_2$};
    \node[state, right of = p2] (p3) {$q_3$};
    \node[state, accepting, right of = p3] (p4) {$q_4$};
    \node[state, below right of = p1] (p5) {$q_5$};
    \node[state, right of = p5] (p6) {$q_{\text{\scalebox{0.75}{$\varnothing$}}}$};
    \node[fill=white, below right=2.1cm and 0.2cm of p4]{$\automaton$};
    
    \draw (p1) edge[loop below] node{$a, b$} (p1)
          (p1) edge[above left] node{$a$} (p2)
          (p1) edge[above right] node{$b$} (p5)
          (p2) edge[above] node{$a$} (p3)
          (p3) edge[above] node{$a$} (p4)
          (p4) edge[loop below] node{$a, b$} (p4)
          (p2) edge[right] node{$b$} (p6)
          (p3) edge[right] node{$b$} (p6)
          (p5) edge[above] node{$a, b$} (p6)
          (p6) edge[loop right] node{$a, b$} (p6);
    
    \node[state, accepting, above right = 3.5cm and 3.8cm of p1] (q1) {$p_1$};
    \node[state, above right = {12mm} and {3mm} of q1] (q2) {$p_2$};
    \node[state, above right of=q2] (q3) {$p_3$};
    \node[state, right of=q3] (q4) {$p_4$};
    \node[fill=white, below right of=q4] (q?) {$\ldots$};
    \node[state, right = {43mm} of q1] (qn) {$p_t$};
    \node[fill=white, below right=0.6cm and 0.1cm of qn]{$\mfautomaton$};
    \draw (q1) edge[in=85,out=115,loop, above] node{$b$} (q1)
          (q2) edge[loop above] node{$b$} (q2)
          (q3) edge[loop above] node{$b$} (q3)
          (q4) edge[loop above] node{$b$} (q4)
          (qn) edge[loop right] node{$b$} (qn)
          (q1) edge[above left] node{$a$} (q2)
          (q2) edge[above left] node{$a$} (q3)
          (q3) edge[above] node{$a$} (q4)
          (q4) edge[above right] node{$a$} (q?)
          (q?) edge[above right] node{$a$} (qn)
          (qn) edge[above] node{$a, b$} (q1)
          (q2) edge[bend left=1cm,left] node{$b$} (q1)
          (q3) edge[bend left=1cm,left] node{$b$} (q1)
          (q4) edge[bend left=1cm,above] node{$b$} (q1)
          ;
    
    \draw (p1) edge[left, bend left=1.3cm] node{$\#$} (q1)
          (p2) edge[left, bend left=1cm] node{$\#$} (q1)
          (p3) edge[left, bend left=0.4cm] node{$\#$} (q1)
          (p5) edge[left, bend left=0.7cm] node{$\#$} (q1)
          (p6) edge[left, bend left=0.1cm] node{$\#$} (q1);
    
    \draw (p4) edge[above right, bend right=0.1cm] node{$\#$} (q1)
          (p4) edge[right, bend right=0.4cm] node{$\#$} (q2)
          (p4) edge[right, bend right=0.7cm] node{$\#$} (q3)
          (p4) edge[right, bend right=0.9cm] node{$\#$} (q4)
          (p4) edge[below right, bend right=0.7cm] node{$\#$} (qn);
              
    \node[draw, dashed, inner xsep=6mm, inner ysep = 4mm, fit=(p1) (p2) (p3) (p4) (p5) (p6)] {};
    \node[draw, dashed,inner xsep=8mm, inner ysep = 10mm, fit=(q1) (q2) (q3) (q4) (q?) (qn)] {};
    \end{tikzpicture}
    }
    \caption{The construction of the automaton $\automaton'$ from an arbitrary NFA $\automaton$ with an explicit dead state $q_\varnothing$ and the $n(k+1)$-state Meyer--Fischer automaton $\mfautomaton$.
    All non-final states of $\automaton$ are connected to the initial state $p_1$ of $\mfautomaton$ using symbol $\#$ and all final states are connected to all states of $\mfautomaton$ using the same symbol $\#$.
    In the diagram, $t = n(k+1)$.
    }
    \label{fig:subset-pspace}
    \end{figure}

    Next, we create an automaton $\automaton'$ over the alphabet $\alphabet' = \{\syma, \symb, \#\}$ by conjoining $\automaton$ and $\mfautomaton$ as follows (see \cref{fig:subset-pspace}).
    \begin{enumerate}
        \item We add $\#$-labeled transitions from all nonfinal states of $\automaton$ to $\mfautomaton$'s initial state $p_1$.
        The state $p_1$ is no longer initial in $\automaton'$.
        \item We add $\#$-labeled transitions from all final states in $\automaton$ to all states of $\mfautomaton$.
    \end{enumerate}
    Next, we analyze an application of the subset construction to $\automaton'$ by breaking its computation into two stages.
    In the first stage, we only traverse transitions with symbols $\syma$ and $\symb$.
    This corresponds \emph{exactly} to running the subset construction on $\automaton$.
    After the first stage is completed, i.e., after no new states can be generated, we allow the subset construction to consider other symbols from $\alphabet'$. 
    Note that our decomposition into two stages is little more than an enforcement of a specific ordering of the computation performed by the subset construction when determinizing $\automaton'$---it does not change \emph{what} computation performed.
    We denote the resulting subset automaton as $\ssautomatonprime$.\looseness=-1

    Let us observe the (incomplete) subset automaton after the first stage. 
    We first consider the case where $\automaton$ is universal.
    Due to \cref{claim:all-states}, all states in the (incomplete) subset automaton contain at least one final state. 
    Hence, in the second stage of the subset construction, any $\#$-labeled transition leads to the state $\powerstate_1$, where, due to \cref{claim:MF-fixpoint}, we either remain forever or transition into the state corresponding to $\emptyset$. 
    This means that exactly 2 additional states are produced in the second stage.  
    We next consider the case where $\automaton$ is \emph{not }universal. 
    In this case, due to \cref{claim:all-states}, the subset construction will have produced some state $\powerstate$ such that $\powerstate \cap \final =  \emptyset$.
    Moreover, because we have an explicit dead state, we know that every $\powerstate$ will be non-empty. 
    Thus, there exists some non-final $\stateq \in \powerstate$, and, in the second stage, the subset construction will take the $\#$-labeled transition from $\powerstate$, generating the state $\powerstate' = \{ p_1 \}$. 
    Subsequently, the subset construction will begin to determinize the $n(k+1)$-state Meyer--Fischer automaton $\mfautomaton$, which produces $2^{n(k+1)}$ additional states.
    
    In summary, if $\automaton$ is universal, then the subset automaton $\ssautomatonprime$ will have exactly $2$ more states than $\ssautomaton$, while if $\automaton$ is not universal, it will have at least $2^{n(k+1)}$ more states. 
    Moreover, since $\automaton$ has only $n$ states, the size of $\ssautomaton$ will be at most $2^{n}$. Thus, if $\automaton$ is universal, we have $|\ssautomatonprime| \leq 2^n + 2$. 
    And, if $\automaton$ is not universal, we have $|\ssautomatonprime| \geq 2^{n(k+1)}$. 
    It follows that if $\automaton$ is not universal, any $p$-approximation $\widehat{m}$ (with $p$ having degree $k$) for $\automaton'$ will satisfy
    \begin{equation}
        \widehat{m} \leq p(m + |\automaton'|) \in \bigO{{(m+n(k+2))}^k} \leq \bigO{{(2^{n} + 2 + n(k+2))}^k} < 2^{n(k+1)} 
    \end{equation}
    for sufficiently large values of $n$.
    And, if $\automaton$ is universal, $\widehat{m} \geq 2^{n(k+1)}$.
    By comparing $\widehat{m}$ and $2^{n(k+1)}$, we are able to decide the universality of $\automaton$. 
    An analogous proof holds for $p$-blow-up.
    \blackqed
\end{proof}

While we believe the next result to be known, we were unable to find a formal presentation or reference for it in the literature, hence we state it as an observation and provide a proof sketch.
\begin{observation}
  Given an NFA $\automaton$, determining the state complexity and the size of the subset automaton can be computed in polynomial space. 
  Hence, the corresponding decision problems are in $\pspace$.
\end{observation}

\begin{proof}[sketch] Consider the following two decision problems.
(i) For a given automaton $\automaton$ and $k \in \N$, is the size of the subset automaton larger than $k$?  
(ii) For a given automaton $\automaton$ and $k \in \N$, is the state complexity of $\automaton$ larger than $k$? 
We claim that both of these problems are in $\npspace$, and, hence, due to Savitch's theorem, also in $\pspace$ \cite{savitch}.\looseness=-1

First, let $\automaton = \fsatuple$ be an NFA. 
Consider the set of characteristic vectors for the set of states $\states$ in the NFA $\automaton$.
Additionally, consider a $\nstates$-bit counter which stores the number of distinct accessible states in the corresponding subset automaton.\looseness=-1

To demonstrate that (i) is in $\npspace$, we enumerate the $2^{|\states|}$ characteristic vectors in an arbitrary total ordering.
For each vector in the enumeration, consider the corresponding subset of the states $\powerstate \subseteq \states$.
Nondeterministically guess the path from $\powerstate_\initial$ (defined in \Cref{alg:subsetconstruction}) to $\powerstate$ if one exists. 
The nondeterministic guessing procedure has at most $2^{\nstates}$ steps, where, in each step, we either nondeterministically choose a symbol $\sym \in \alphabet$ that leads us to the next state or halt.
In the case that we successfully guessed a path to $\powerstate$, we increase the $\nstates$-bit counter by $1$; otherwise, we keep it unchanged. 
After repeating this for all $2^\nstates$ characteristic vectors, we compare the counter value and $k$, i.e., accept or reject the input.
Thus, (i) is in $\npspace$ and, by Savitch's theorem, in $\pspace$.
    
To demonstrate that (ii) is in $\npspace$, we borrow an idea suggested by Stefan Kiefer \cite{pspace-minimization}.
Because DFA minimization is in $\mathsf{NL}$ \cite{NL-minimization} and, hence, due to Savitch's theorem, in $\mathsf{polyL}$, we are only required to keep track of polynomially many states from subset construction at once to compute the number of equivalence classes, i.e., determine the state complexity.
    \blackqed
\end{proof}

\section{Bounding the State Complexity}\label{sec:subset-complexity}
Due to \cref{prop:pspace} and \cref{prop:subset-pspace}, no reasonable approximation to the state complexity or the size of the DFA produced by the subset construction is possible assuming $\pspace \neq \pclass$.
Thus, an alternative is to look for some easily computable lower and upper bounds and keep them as general as possible.
We focus on bounding the size of the DFA produced by the subset construction (\cref{alg:subsetconstruction}), which directly tells us how hard it would be to determinize the chosen NFA in practice.
Note, however, that, by the results in \Cref{sec:difficulties}, we repeat that such bounds cannot be both (i) arbitrarily good, and (ii) easy to compute for all NFAs.\looseness=-1

One way of achieving such a bound for the subset construction is to consider the transition monoid $\transmonoid{\automaton}$.
It is known that the size of the transition monoid directly translates into an upper bound on the subset automaton. 
In the case of unary, commutative, and dense automata, this fact can be used to produce useful bounds \cite{onthefly}. 
However, it remains unclear how to bound the size of the transition monoid for an arbitrary NFA.

\begin{lemma}[Monoid bound \cite{onthefly}]
    Given an NFA $\automaton$, the DFA produced by the subset construction \textup{$\ssautomaton$} satisfies \textup{$|\ssautomaton| \leq |\transmonoid{\automaton}|$}.
\end{lemma}

Another natural way to bound the size of subset automaton is to consider the range $\range{T^{(\sym)}}$ for all transition matrices in 
\begin{equation}\label{eq:transition-matrices}
\sT \coloneqq \{ T^{(\sym)} \ | \ \sym \in \alphabet \}.
\end{equation}
Because every state in the subset construction corresponds to a characteristic vector over $\states$, it has to be in the image of some transition matrix. 
Hence, the size of their combined images is an upper bound on the size of the subset automaton.

\begin{lemma}[Range bound]\label{prop:range}
   Let $\automaton$ be an NFA with a set of transition matrices $\sT$. 
   Then the DFA output by the subset construction $\ssautomaton$ satisfies
    \textup{
    \begin{equation}
        |\ssautomaton| \leq 1 +\sum_{\sym \in \Sigma} \left|\range{T^{(\sym)}} \right|,
    \end{equation}
    }
     where $\range{T}$ is the range of the matrix $T$ over the Boolean semifield $\B$. 
\end{lemma}

\begin{proof}
    Let $\ssautomaton = (\alphabet, \ssstates, \initial, \ssfinal, \sstrans)$. 
    Consider any $\powerstate_J \in \ssstates$
    and let $\sym_1 \cdots \sym_N \in \kleene{\alphabet}$ be a word such that $\powerstate_J = \trans(\powerstate_I, \sym_1 \cdots \sym_N)$
    Then, $\powerstate_J$'s characteristic vector is given by
    \begin{equation}
        J = I^\top T^{(\sym_1)} T^{(\sym_2)} \cdots T^{(\sym_N)}.
    \end{equation}
    Moreover, because this holds for any $J$, we have 
    \begin{equation}
        \ssstates \setminus \{ \powerstate_\initial \} \subseteq  \bigcup_{\sym \in \alphabet} \range{T^{(\sym)\top}} =  \bigcup_{\sym \in \Sigma} \range{T^{(\sym)}}.
    \end{equation}
    Applying the union bound and adding $1$ to account for the initial state $\powerstate_\initial$ concludes the proof.
    \blackqed
\end{proof}

Serendipitously, both the range and monoid bounds can be naturally combined, admitting a more fine-grained perspective on the nondeterministic blow-up, each compensating for the slackness of the other. 
\begin{definition}\label{def:subset-complexity}
   Given an NFA $\automaton = \fsatuple$, we define the \defn{subset complexity} \textup{$\norm{\automaton}$} as
    \begin{equation}
        \norm{\automaton} = \min_{\sJ \subseteq \alphabet} \left(1 + \sum_{\sym \in \alphabet \setminus \sJ} \left|\range{T^{(\sym)}} \right| \right) |\transmonoid{\sJ}|,
    \end{equation}
    where $\transmonoid{\sJ}$ denotes the submonoid generated by transitions $\{ T^{(\sym)} \ | \ \sym \in \sJ \}$.\looseness=-1
\end{definition}
Essentially, subset complexity allows us to split $\automaton$ into subautomata according to whether the symbols on the transitions are in $\sJ$ or $\Sigma \setminus \sJ$, respectively.
Then, we either apply the monoid or the range bound independently.
The subset complexity is then the best possible split.

\begin{theorem}[Subset complexity]\label{thm:decomposition-bound}
  Let $\automaton$ be an NFA and let \textup{$\ssautomaton$} be the DFA output by the subset construction with $\automaton$ as input.
  Then, we have \textup{$|\ssautomaton| \leq \norm{\automaton}$}.
\end{theorem}
\begin{proof}
    Let $\sJ \subseteq \alphabet$ be arbitrary. 
    We first show that
    \begin{equation}
        |\ssautomaton| \leq \left(1 + \sum_{a \in \Sigma \setminus \sJ} \left|\range{T^{(\sym)}} \right| \right) |\transmonoid{\sJ}|.
    \end{equation}
    For any word $\strx \in \kleene{\alphabet}$,
    let $\strx = \stry \strz = \stry z_1 \cdots z_M$ be the unique decomposition where $z_m \in \sJ$ and $M$ is maximal.
Furthermore, let $\powerstate_J = \trans(\initpowerstate, \strx)$.
  Similar to the proof of \cref{prop:range}, we consider a characteristic vector $J$ for $\powerstate_J$.
    For clarity, we will denote transition matrices corresponding to symbols $\sym \in \sJ$ with $M^{(\sym)}$ and transition matrices which correspond to symbols $\sym \notin \sJ$ with $R^{(\sym)}$.
    Then, the characteristic vector $J$ can be written as
    \begin{equation}
        J = J' M^{(z_{1})} \cdots M^{(z_{M})},
    \end{equation}
    for some $J' \in \bigcup_{\sym \in \alphabet \setminus \sJ} \range{R^{(\sym)}}$.\footnote{Note that $J'$ will, in general, be the product of both $M$ and $R$ matrices.}
    Thus, there are at most $1 + \sum_{\sym \in \Sigma \setminus \sJ} \left|\range{R^{(\sym)}} \right|$ choices for a vector $J'$ and $|\transmonoid{\sJ}|$ choices for the matrix $M^{(z_{1})} \cdots M^{(z_{M})}$.
    Thus, the total number of choices for $J$ is bounded by the corresponding product. The minimum bound follows.\looseness=-1
    \blackqed
\end{proof}

Our notion of subset complexity is a direct improvement over both monoid and range bounds, because it allows us to cherry-pick individual transition matrices according to their properties to achieve the tightest bound. 

\begin{example}
    Consider the $n$-state Moore automaton $\mooreautomaton$ from \cref{fig:Moore-NFA} with state complexity of $2^n$.
    We define a modified $n$-state Moore automaton $\mooreautomatonprime$ as follows.
    Take the $n$-state Moore automaton and replace the labels on the transitions $\edgenoweight{\stateq_n}{c}{\stateq_1}$ and $\edgenoweight{\stateq_n}{c}{\stateq_2}$ with a new symbol $\symc \notin \{a, b\}$. The new automaton $\mooreautomatonprime$ operates over the alphabet $\alphabet = \{\syma, \symb, \symc \}$; it is still nondeterministic.

    \begin{claim}
        For an $n$-state modified Moore automaton $\mooreautomatonprime$, the following inequality holds $\norm{\mooreautomatonprime} \leq 3n^2 + 3n$.
    \end{claim}
    
    \begin{proof}
        Let $\gJ = \{\syma, \symb\}$ and observe that $T^{(\syma)} T^{(\symb)} = T^{(\symb}T^{(\syma)}$, which means that all matrices in the transition monoid $\transmonoid{\gJ}$ can be obtained from a product sequence of the following form
        \begin{equation*}
            T^{(\symb)}T^{(\symb)}\cdots T^{(\symb)}T^{(\syma)}T^{(\syma)}\cdots T^{(\syma)}.
        \end{equation*}
        Moreover, we have that ${(T^{(\syma)})}^n = 0$, i.e., the zero matrix and ${(T^{(\symb)})}^{n-1} = {(T^{(\symb)})}^n$, meaning that there are at most $n + 1$ distinct powers of $T^{(\sym)}$ and $n$ distinct powers of $T^{(\symb)}$ (counting the empty set). 
        Hence, $|\transmonoid{\sJ}| \leq n(n+1) = n^2 + n$.
        Moreover, we have that $\gR(T^{(\symc)}) = 2$ because $T^{(\symc)}$ is a rank $1$ matrix. 
        Applying \cref{thm:decomposition-bound} concludes the proof.\looseness=-1 
        \blackqed
    \end{proof}
\end{example} 

\section{A Simple Upper Bound}\label{sec:computational}

In this section, we discuss how to efficiently compute the estimates for the subset complexity from \cref{def:subset-complexity}. This requires computing the size of the range space as well as the size of the transition monoid. The size of the range space is known to be equivalent to counting maximal independent sets in bipartite graphs \cite{rangeduality}, which is a $\#\pclass$-complete problem \cite{countingmaximalis}. The monoid size is generally hard to determine.\footnote{One notable exception is when all the matrices are permutations, in which case we use the Schreier--Sims algorithm to compute the size in almost quadratic time \cite{schreier-sims}.}
Despite these obstacles, we can still derive an efficiently computable upper bound for subset complexity. We start by considering the case where our NFA is unary: for some automaton $\automaton_1$ let $A$ be its unique transition matrix. Then the bound on $|\transmonoid{\automaton_1}|$ relies solely on the \defn{precedence graph} $\graph(A)$, i.e., a directed graph with $|\automaton|$ vertices and edges defined by the transition matrix $A$.

\begin{definition}\label{def:cyclicity}
    Given a directed graph $\mathcal{\graph}$ the \defn{cyclicity} of $\graph$, denoted with $c(\graph)$, is the least common multiple of cyclities of all maximal strongly connected components of $\graph$, where the cyclicity of a strongly connected component $\sC$ is the greatest common divisor of all cycles in $\sC$.
\end{definition}

\begin{lemma}[Section 3 in \cite{onthefly}]\label{lem:monoid-one-letter}
    Let $\automaton_1 = (\{ \sym \}, \states_1, \trans_1, \initial_1, \final_1)$ be a $n$-state unary automaton with transition matrix $A$. Then 
    \begin{equation}
        c(\graph(A)) \leq |\transmonoid{\automaton_1}| \leq  c(\graph(A)) + n^2 - 2n + 2.
    \end{equation}
\end{lemma}

Now we turn our attention to the range of a Boolean matrix and the methods that can be used to bound its size.

\begin{lemma}[Lemma 1 in \cite{boolean-width}]\label{lem:range-bounds}
    Given an $n \times n$ Boolean matrix $A$, let $\rank{A}$ be the rank of $A$ over the Galois field $\field_2$. Then the range of $A$ satisfies
    \begin{equation*}
        \rank{A} \leq |\range{A}| \leq 2^{\rank{A}^2/4 + \bigO{\rank{A}}}.
    \end{equation*}
\end{lemma}

\begin{remark}
    Other possibilities for bounding the range exist.
    For example, one can consider the size of maximum matching or the notion of Boolean width; see \cite{boolean-width} for in-depth discussion.
\end{remark}

Using the simplified bounds for range on symbols $\alphabet \setminus \{ \syma \}$ and monoid size on the remaining symbol $\syma$ we get a simple upper bound on the subset complexity for any automaton $\automaton$. 
\begin{proposition}[All-but-one]\label{prop:all-but-one}
    Given an $n$-state NFA $\automaton = \fsatuple$, an arbitrary target symbol $\syma \in \alphabet$ and any $\varepsilon > 0$, the following holds
    \begin{equation*}
        \norm{\automaton} \in \bigO{\nsymbols(c(\graph(T^{(\syma)})) + n^2)  \max_{\symb \in \alphabet \setminus \{ \syma \}}2^{\rank{T^{(\symb)}}^2/(4 - \varepsilon)}}.
    \end{equation*}

\end{proposition}
This bound has an intuitive interpretation---if all but one of the transition matrices are of low rank, then the subset complexity will be polynomially bounded; the cyclicity is large only in very exceptional cases. 
Note that cyclicity can be computed using a simple breadth-first search, while rank can be calculated using Gaussian elimination over $\field_2$.

\section{Conclusion}
Given an NFA, we have attempted to forecast the state complexity of the minimal equivalent DFA as well as the size of the subset automaton. 
In both cases, we show that such forecasting is $\pspace$-hard; in fact, even the approximate forecasting is $\pspace$-hard.
As a potential workaround, we propose subset complexity---a tractable algebraic criterion, which gives us a sufficient condition for when determinization can be performed efficiently.\looseness=-1


\bibliographystyle{splncs04}
\bibliography{main}





\end{document}